\documentclass[letterpaper, USenglish]{lipics-v2016}

\usepackage{amsmath}
\usepackage{amsfonts}
\usepackage{amsthm}
\usepackage{amssymb}
\usepackage[utf8]{inputenc}

\usepackage{textcomp}
\usepackage{cleveref}

\usepackage{wrapfig}

\usepackage{indentfirst}

\usepackage{pbox}
\usepackage{comment}

\usepackage[linesnumbered, boxruled]{algorithm2e}
\usepackage{pgfplots}
\pgfplotsset{compat=1.4}
\usetikzlibrary{shapes}
\usetikzlibrary{plotmarks}

\usepackage{framed}
\usepackage{mdframed}
\usepackage{placeins}
\usepackage{subcaption}
\usepackage{afterpage}
\usepackage{hyperref}
\usepackage{courier}
\usepackage{complexity}

\newcommand{\eps}{\varepsilon}
\newcommand{\Oish}{\widetilde{O}}

\newcommand{\tee}{\mathcal{T}}

\newcommand{\ba}{\mathbf{a}}
\newcommand{\bd}{\mathbf{d}}
\newcommand{\bD}{\mathbf{D}}
\newcommand{\be}{\mathbf{e}}
\newcommand{\beff}{\mathbf{f}}
\newcommand{\bu}{\mathbf{u}}

\newcommand{\bv}{\mathbf{v}}
\newcommand{\bp}{\mathbf{p}}

\newcommand{\bs}{\mathbf{s}}
\newcommand{\bt}{\mathbf{t}}
\newcommand{\bx}{\mathbf{x}}
\newcommand{\by}{\mathbf{y}}

\newcommand{\rr}{\mathbb{R}}
\newcommand{\bz}{\mathbf{0}}
\newcommand{\bo}{\mathbf{1}}
\newcommand{\rrp}{\mathbb{R}_{\ge \bz}}

\newtheorem{claim}{Claim}

\theoremstyle{definition}

\newtheorem*{definition*}{Definition}
\newtheorem*{theorem*}{Theorem}
\newtheorem*{mtheorem*}{Main Theorem}

\newtheorem*{corollary*}{Corollary}
\newtheorem*{lemma*}{Lemma}
\newtheorem*{mlemma*}{Main Lemma}
\newtheorem{program}{Program}


\title{Testing Core Membership in Public Goods Economies}
\titlerunning{Testing Core Membership in Public Goods Economies}
\author{Greg Bodwin}
\affil{MIT EECS}
\date{}

\begin{document}
\maketitle

\begin{abstract}
\indent This paper develops a recent line of economic theory seeking to understand public goods economies using methods of topological analysis.
Our first main result is a very clean characterization of the economy's core (the standard solution concept in public goods).
Specifically, we prove that a point is in the core iff it is Pareto efficient, individually rational, and the set of points it dominates is path connected.

While this structural theorem has a few interesting implications in economic theory, the main focus of the second part of this paper is on a particular algorithmic application that demonstrates its utility.
Since the 1960s, economists have looked for an efficient computational process that decides whether or not a given point is in the core.
All known algorithms so far run in exponential time (except in some artificially restricted settings).
By heavily exploiting our new structure, we propose a new algorithm for testing core membership whose computational bottleneck is the solution of $O(n)$ convex optimization problems on the utility function governing the economy.
It is fairly natural to assume that convex optimization should be feasible, as it is needed even for very basic economic computational tasks such as testing Pareto efficiency.
Nevertheless, even without this assumption, our work implies for the first time that core membership can be efficiently tested on (e.g.) utility functions that admit ``nice'' analytic expressions, or that appropriately defined $\eps$-approximate versions of the problem are tractable (by using modern black-box $\eps$-approximate convex optimization algorithms).
\end{abstract}

\section{Introduction}

\subsection{Background on Public Goods Economics}

A basic question in economics is to understand the forces governing the production of public goods.
A good is \emph{public} if its use by one person does not reduce its availability to others, and if none are excluded from using the good.
Examples include public parks, research information, a clean environment, national defense, radio broadcasts, and so on.

Public goods economies were first explicitly abstracted in a classic paper by Samuelson in 1954 \cite{Samuelson54}, and have since become central objects of study for economists.
An important feature of public goods economies is that they are not well modeled by the individualistic ``best-response dynamics'' which govern familiar economic equilibrium concepts such as the Nash or Walrasian equilibrium.
Rather, public goods typically arise as the result of a communal process -- e.g. negotiations, treaties, or taxes -- that allow the cost of production to be amortized over all agents that stand to benefit from the good.
Accordingly, public goods economics inspired the development of \emph{cooperative game theory}, which seeks to understand these cooperative dynamics and when an agreement to produce public goods is ``stable'' or when it is doomed to fall apart.

What, exactly, does ``stability'' mean in this context?
The most standard notion is \emph{coalitional stability}, which is given as follows:

\begin{definition*} [Informal]
Let $\ba$ be an outcome in a public goods economy with agents $N$ (i.e. $\ba$ describes the amount of work each agent contributes to produce a public good).
We say that $\ba'$ is a \emph{deviation} on $\ba$ for a nonempty coalition of agents $C \subseteq N$ if no agents in $N \setminus C$ perform work (i.e. $a'_i = 0$ for all $i \notin C$) and all agents in $C$ prefer $\ba'$ to $\ba$.
If there is no deviation on $\ba$ for any coalition, then we say $\ba$ is \emph{coalitionally stable}.
The set of coalitionally stable outcomes is called the \emph{core} of the economy.
\end{definition*}

Since its inception in the late 1800's \cite{Edgeworth81},\footnote{Some of the early work on cooperative game theory used the name \emph{contract curve} instead of core.  The term \emph{core} was coined in \cite{Gillies59}.} the core and cooperative game theory in general have played major roles in many successful economic research programs.
More background can be found in most modern game theory textbooks, e.g. \cite{PS07, Curriel97}.

\subsection{(Non-)Algorithmic Properties of Public Goods Economies}

An inherent conceptual drawback of coalitional stability is its exponential-size definition.
In other words, for an outcome $\ba$ to be coalitionally stable, \emph{every single one} of the $2^n -1$ possible coalitions of agents must not have a deviation.
Hence, the naive algorithms for testing the coalitional stability of $\ba$ must perform computations for all $2^n-1$ coalitions, and so they suffer exponential runtime.
Coalitional stability is not a very convincing solution concept if its implicit notion of ``instability'' assumes that agents can quickly make exponential time computations in order to find deviations whenever they exist.




This problem has led economists studying public goods to seek more clever methods for solving computational problems related to the core of a public goods economy, which avoid this exponential behavior.
Some of the initial work in this vein attacked the closely-related problem of simply outputting any core outcome.
The first such solution appeared in the 1960s, when Scarf \cite{Scarf67a} proved that ``balanced'' games have nonempty cores, by means of an (exponential time) algorithm that outputs a core point and provably always terminates.
Similar results were proved in the setting of public goods economies by Chander and Tulkens \cite{CT97} and Elliot and Golub \cite{EG13}.
Meanwhile, followup work has suggested that the slow runtime of Scarf's algorithm may be inherent to the problem: Kintali et al \cite{KPRST09} showed that Scarf's algorithm cannot be improved to polynomial runtime (unless P = PPAD), and Deng and Papadimitriou \cite{DP94} showed that it is NP Complete just to detect whether or not the core is empty (let alone find a point in the core), even in the simple class of graphical games, and even when Scarf's assumption of ``balance'' is dropped.
Other notable hardness results in this vein have come from Conitzer and Sandholm \cite{CS06} and Greco et al \cite{GMPS11}.

There has also been considerable prior work on the ``membership testing'' problem of determining whether a point taken on input is in the core (this is the question addressed in this paper).
Deng and Papadimitriou's work \cite{DP94} also implies that membership testing is NP hard even in the restricted setting of graphical games, although there are straightforward efficient algorithms in the further restricted setting where the game is superadditive.
Conitzer and Sandholm \cite{CS04} showed that membership testing is co-NP complete in games where coalition values have a ``multiple-issue'' representation in polynomial space.
Faigle et al \cite{FFHK97} showed that the problem is NP complete in a variant of graphical games where payoffs are given by minimum spanning trees of subgraphs.
Sung and Dimitrov \cite{SD07} showed co-NP completeness for membership testing in ``hedonic coalition formation games.''
Goemans and Skutella showed NP completeness for both emptiness and membership testing in ``facility location games,'' and gave formulations of these problems as LP relaxations \cite{GS00}.
There has been work on games defined by \emph{marginal contribution nets} (MC-nets) \cite{IS05, EGGW09}, in which values attainable by coalitions are determined by succinct logical formulae.
Li and Conitzer \cite{LC14} studied emptiness testing and membership testing under various classes of formulae, and obtained various algorithms or NP hardness results depending on the complexity of the formulae allowed.


\subsection{Our Results}


A recently popular trend in public goods research has been to model economies as networks, and then seek to analyze the economy by studying the topological properties of the underlying ``benefits network,'' describing the ability of agents to transfer utility to each other at any given point (see for example \cite{BCZ06, BK07, Allouch13, Allouch15, EG13}).
Much of the initial work focused on Nash equilibria \cite{BCZ06, BK07, Allouch13, Allouch15} of the economy.
A major stride was recently taken by Elliot and Golub \cite{EG13}, who extended the theory to show that the Lindahl equilibria\footnote{The Lindahl equilibria of an economy are the competitive equilibria that would be reached if market externalities were truthfully reported and then bought and sold on an open market.  A formal definition is not necessary to read this paper, but can be found in e.g. \cite{PS07, Curriel97}.} of an economy are precisely the points that are eigenvalues of their own benefits network.
They further discuss connections between the Lindahl equilibria and the core -- in particular, in any standard model (including theirs) all Lindahl equilibria are also core outcomes \cite{Foley70}.
However, they raise an interesting open question to more precisely characterize the core \cite{GolubPC} in a similar vein.

Our first main result achieves this goal.
We show that the core can be characterized as follows:
\begin{theorem} \label{thm:main}
Let $\ba$ be an outcome in a public goods economy and let $\bD_{\ba}$ be the set of points that no agent prefers to $\ba$.
Then $\ba$ is in the core if and only if it is Pareto efficient, individually rational, and $\bD_{\ba}$ is path connected.
\end{theorem}
(Here, the definition of path connectedness is the standard topological one: for any two points $\bx, \by \in \bD_{\ba}$, there is a continuous function $f : [0, 1] \to \rrp^n$ satisfying $f(0) = \bx, f(1) = \by,$ and $f(\lambda) \in \bD_{\ba}$ for all $0 \le \lambda \le 1$.  The image of $[0, 1]$ under $f$ is called a path.)

An interesting consequence of this theorem is a precise description of the relationship between Lindahl equilibria and core outcomes:
\begin{theorem}
Assuming that the utility function $\bu$ is differentiable, the Lindahl equilibria of a public goods economy are precisely the core points whose core membership can be certified using only local information.
\end{theorem}
The proof of this corollary is essentially immediate by combining Theorem \ref{thm:main} with a more technical phrasing of Elliot and Golub's result.

While we believe that these two structural theorems hold intrinsic interest, the second half of this paper is intended to demonstrate their power by an application to the algorithmic problem discussed earlier.
We have previously suggested the intuition that the algorithmic core membership testing problem is hard because the naive algorithms must check an exponential number of coalitions for a potential deviation.
However, Theorem \ref{thm:main} lets us avoid this brute-force behavior: after checking for Pareto efficiency and individual rationality (which is quite easy), we are left only with the task of checking whether or not $\bD_{\ba}$ is path-connected.
The complexity of this task is non-obvious, but we show that it can be done fairly efficiently, yielding the following result:

\begin{theorem} \label{thm:alg}
Given an outcome $\ba$ in an $n$-agent public goods economy, there is an algorithm that decides whether or not $\ba$ is in the core of a public goods economy.
The computational bottleneck in this algorithm is the solution of $O(n)$ convex programming problems on the utility function of the economy.
\end{theorem}
Hence we essentially have the first polynomial-time tester for coalitional stability in an unrestricted public goods economy, up to the implementation of the necessary convex programming oracle.
It is fairly natural in our economic metaphor to assume that convex programming should be tractable: it corresponds to the negotiation process of a group of agents trying to determine how well they can maximize a joint utility function as a group.
If even this is impossible, then it is essentially hopeless to efficiently test core membership. One cannot even test the more basic property of Pareto efficiency, a necessary step towards testing core membership, without assuming some computational power along these lines.

Even so, if one does not wish to introduce such assumptions, Theorem \ref{thm:alg} implies that several broad special cases of membership testing have efficient algorithms.
The most obvious of these is when the utility function and its derivative can be described by a ``nice'' analytic function on which the standard derivative-based method for exact convex optimization goes through.
Less obviously, if convex optimization really is hard for the given utility function (or the utility function of the economy is unknown), one can employ modern $\eps$-approximate convex optimization solvers, which treat the utility function as a black box that can be queried, to solve certain natural $\eps$-approximate relaxations of the core membership testing problem.
We discuss this point in the conclusion of the paper, since it is easier to be specific here once the economic model is familiar.

We consider it somewhat surprising that these algorithmic results are possible, given a general dearth of positive results in the area.
Moreover, the approach taken by the algorithm is fairly intuitive and seems to plausibly reflect practical behavior.
Starting with the grand coalition, we show (via Theorem \ref{thm:main}) that we can either determine that the current coalition has a deviation, or we can identify a ``least-valuable player'' who is formally the least likely agent to participate in a deviation.
We then kill this agent and repeat the analysis on the survivors.
After $n$ rounds, we have either killed every agent (and thus determined that the given point is coalitionally stable), or we have explicitly found a surviving coalition with a deviation.
It is quite reasonable to imagine that a practical search for a deviating coalition might employ a ``greedy'' method of iteratively killing the agent who seems to be least pulling their weight at the current agreement; an insight of Theorem \ref{thm:alg} is that this search heuristic is in fact thorough and will provably produce the right answer.


\subsection{Comparison with Prior Work.}

Elliot and Golub \cite{EG13} recently studied the Lindahl equilibria in public goods economies, with a focus on characterizing the set of solutions rather than algorithmically computing/testing them.
More specifically, they frame the typical model of public goods economies in the language of networks, and use this to equate the eigenvectors of the ``benefits network'' with the Lindahl equilibria of the economy.
A less general version of this networks interpretation was implicitly used in several other papers concerning Nash equilibria of public goods economies, for example \cite{BCZ06, BK07, Allouch13, Allouch15}.
In this paper, we will adopt the more general networks-based phrasing of public goods economies used by Elliot and Golub, and we will rely on this insight in a critical way to prove our main results.

Per the discussion above, there has been lots of prior work on the algorithmic properties of the core, largely intended to confirm/refute the bounded rationality argument in some economic model.
Three questions are commonly studied:
\begin{itemize}
\item The \emph{membership testing} problem (discussed above): is a given outcome in the core of the game?
\item The \emph{emptiness testing} problem: is the core empty?
\item The \emph{member finding} problem: output any solution in the core of the game (if nonempty).
\end{itemize}

We remark that the latter two problems are already closed in public goods economies: Elliot and Golub \cite{EG13} show that the core is never empty except in certain degenerate cases, and it can be seen from the model below that the member finding problem is essentially identical to the general problem of convex optimization (which is well beyond the scope of this economically-minded research program).
Hence, this work is entirely focused on the membership testing problem.

In order to frame these three questions as proper computational problems, past work has commonly defined a ``compressed'' cooperative game that allows the payoffs achievable by all $2^n$ possible coalitions to be expressed on only $\poly(n)$ input bits.
For example, in a seminal paper by Deng and Papadimitriou introducing this line of research \cite{DP94}, the authors studied \emph{graphical games} in which weighted edges are placed between agents and the value attainable by a coalition is equal to the total weight contained in its induced subgraph.
Upper and lower bounds are often obtained for these problems by exploiting particular features of the compression scheme.
By contrast, our goal is to assume as little structure for the problem as possible (since our main results are upper bounds, this is the more general approach).
Thus, we allow the economy to be governed by an arbitrarily complex utility function, which does not need to have a succinct representation, or even any algorithmic representation at all.
Instead, we allow ourselves black-box constant-time query access to the utility function, which acts as an oracle and thus may have arbitrary complexity.
The goal in this substantial generalization is to ensure that our results reveal structure of the core itself, rather than the nature of an assumed compression.

%
%

\section{The Model and Basic Definitions}

\subsection{Notation Conventions}


Given vectors $\ba, \ba' \in \rr^n$, we will use the following (partial) ordering operations:
\begin{itemize}
\item $\ba \ge \ba'$ means that $a_i \ge a'_i$ for all $1 \le i \le n$,
\item $\ba > \ba'$ means that $a_i > a'_i$ for all $i \le i \le n$, and
\item $\ba \gneq \ba'$ means that $a_i \ge a'_i$ for all $1 \le i \le n$, and $a_j > a'_j$ for some $1 \le j \le n$.
\end{itemize}

Given a subset $C \subseteq \{1, \dots, n\}$ and a vector $\bv \in \rr^n$, we write $\bv_C$ to denote the restriction of $\bv$ to the indices in $C$; that is, $\bv_C$ is the $|C|$ length vector built by deleting the entry $v_i$ from $\bv$ for each $i \notin C$.

We use $\bz, \bo$ as shorthands for the vectors $\langle 0, \dots, 0 \rangle, \langle 1, \dots, 1 \rangle$ respectively.

\subsection{Economic Model}

We adopt the terminology of Elliot and Golub \cite{EG13} when possible.
The salient pieces of our economy are defined as follows:

\begin{itemize}
\item The set of agents in the economy is given by $N = [n] = \{1, \dots, n\}$.
A nonempty subset of agents in the economy $C \subseteq N$ is called a \emph{coalition}.
The coalition $C = N$ is called the \emph{grand coalition}.

\item Each agent $i$ chooses an \emph{action} $a_i$, which can be any real number in the interval $[0, 1]$.
An \emph{outcome} or \emph{point} is a vector $\ba \in \rr^n$ built by concatenating the actions of all agents.

\item There is a continuous \emph{utility function} $\bu : [0,1]^n \to [0, 1]^n$, which maps outcomes to a level of ``utility'' for each agent.
In particular, agent $i$ prefers outcome $\ba$ to outcome $\ba'$ iff $u_i(\ba) > u_i(\ba')$.
The utility function has the following two properties:
\begin{itemize}
\item \emph{Positive Externalities}: whenever $\ba \gneq \ba'$ with $a_i = a'_i$, we have $u_i(\ba) > u_i(\ba')$.
This assumption is what places us in the setting of public goods economies; intuitively, it states that an agent gains utility when other agents increase their production of public goods.

\item \emph{Convex Preferences}: we assume that $\bu$ is concave.\footnote{Confusingly, when $\bu$ is mathematically concave, one says that preferences are ``economically convex'' -- hence, \textit{convex preferences}.}
That is, for any outcomes $\ba, \ba'$ and any $\lambda \in [0, 1]$, we have $\bu(\lambda \ba + (1 - \lambda) \ba') \ge \lambda \bu(\ba) + (1 - \lambda) \bu(\ba')$.
This standard assumption corresponds to the economic principle of diminishing marginal returns.
\end{itemize}
\end{itemize}

\subsection{Game Theory Definitions}

We recap some well-known definitions from the game theory literature.

\begin{definition} [Pareto Efficiency]
An outcome $\ba$ is a \emph{Pareto Improvement} on another outcome $\ba'$ if $\bu(\ba) \gneq \bu(\ba')$.  An outcome $\ba$ is \emph{Pareto Efficient} if there is no Pareto improvement on $\ba$.  The set of Pareto efficient outcomes is called the \emph{Pareto Frontier}.
\end{definition}

The main solution concept that will be discussed in this paper is \emph{the core}:
\begin{definition} [Deviation]
Given an outcome $\ba$, an outcome $\ba'$ is a \emph{deviation} from $\ba$ for a coalition $C$ if $\ba'_{N \setminus C} = \bz$ and $\bu_C(\ba') > \bu_C(\ba)$.
\end{definition}

\begin{definition} [The Core]
An outcome $\ba$ is in the \emph{core} of the economy if no coalition has a deviation from $\ba$ (equivalently, $\ba$ is \emph{coalitionally stable}). 
\end{definition}

%
%

The next definition that will be useful in our proofs is the \emph{projected economy}:
\begin{definition}
Given an economy described by agents $N$ and a utility function $\bu$, the \emph{projected economy} for a coalition $C$ is the economy described by agents $N$ and utility function $\bu^C(\ba_C)$, where
$$\bu^C(\ba_C) := \bu_C(\ba_C \cdot \bz_{N \setminus C}).$$
\end{definition}
In other words, the new $|C|$-dimensional utility function $\bu^C$ is obtained by fixing the actions of $N \setminus C$ at $0$, allowing any action for $C$, and then using the old utility function $\bu$ to determine the utilities for $C$ in the natural way.
We suppress the superscript $\bu^C$ when clear from context.

\begin{definition}
The \emph{dominated set} of $\ba$, denoted $\bD_{\ba}$, is defined as:
$$\bD_{\ba} := \{ \ba' \ \mid \ \bu(\ba) \ge \bu(\ba') \}$$
\end{definition}
In other words, $\bD_{\ba}$ is the set of points that no agent prefers to $\ba$.
Note that this is an unusually weak definition of dominance, in the sense that (for example) $\bD_{\ba}$ contains $\ba$ itself.

\section{A Topological Characterization of the Core}

Our goal in this section is to prove the following structural theorem:

\begin{theorem} \label{thm: core classification}
Let $\ba$ be an outcome in a public goods economy.
Then $\ba$ is in the core if and only if it is Pareto efficient, individually rational, and $\bD_{\ba}$ is path connected.
\end{theorem}

The vast majority of the technical depth of this theorem is tied up in the implication
$$\ba \text{ is in the core } \longrightarrow \bD_{\ba} \text{ is path connected.}$$
The remainder of this forwards implication ($\ba$ is in the core $\to$ neither the grand coalition nor any singleton coalition has a deviation from $\ba$) is extremely straightforward: $\ba$ is individually rational iff each agent $i$ prefers it to the outcome they can guarantee acting alone, which coincides with the notion that the singleton coalition $\{i\}$ has no deviation from $\ba$.
In our model, Pareto efficiency coincides with the notion that the grand coalition $N$ has no deviation from $\ba$:
\begin{claim} \label{clm: vec adjust}
Let $\ba$ be an outcome.
If there is an outcome $\ba'$ satisfying $\ba \gneq \ba'$ ($\ba \lneq \ba'$) and $\bu(\ba) \gneq \bu(\ba')$, then there is an outcome $\ba''$ satisfying $\ba > \ba''$ ($\ba < \ba'$) and $\bu(\ba) > \bu(\ba'')$.
\end{claim}
\begin{proof}
We will prove the claim for the case $\ba \gneq \ba'$; the case $\ba \lneq \ba'$ follows from a symmetric argument.

Choose an agent $i$ for whom $u_i(a) > u_i(a')$, and then slightly increase $a_i$.
Since $\bu$ is continuous, if we increase $a_i$ by a sufficiently small amount then we still have $u_i(a) > u_i(a')$.
Additionally, by positive externalities we then have $\bu(\ba) > \bu(\ba')$.
We can then slightly increase the actions of all agents, such that $\ba > \ba'$, but with sufficiently small increases we do not destroy the property that $\bu(\ba) > \bu(\ba')$.
\end{proof}

In this section, we will first give a complete proof of the (easier) backwards implication of Theorem \ref{thm: core classification}, and then we sketch the proof of the forwards implication.
Due to space constraints, a full proof of the forwards implication can be found in Appendix \ref{app:class proof}.

\subsection{Backwards Implication of Theorem \ref{thm: core classification}}
First:
\begin{lemma} \label{lem: down dev}
If $\ba$ is Pareto efficient, then every deviation $\ba'$ from $\ba$ satisfies $\ba' \lneq \ba$.
\end{lemma}
\begin{proof}
Let $I$ be the set of agents $i$ for which $a'_i > a_i$, and suppose towards a contradiction that $I$ is nonempty.

Consider the point $\ba''$ defined such that $\ba''_{N \setminus I} := \ba_{N \setminus I}$ and $\ba''_I := \ba'_I$.
We then have $\bu_{N \setminus I}(\ba'') > \bu_{N \setminus I}(\ba)$ by positive externalities, since these points differ only in that the (nonempty) coalition $I$ has increased their actions.
We also have $\bu_I(\ba'') \ge \bu_I(\ba') > \bu_I(\ba)$, where the first inequality follows from positive externalities (since these points differ only in that the coalition $N \setminus I$ has weakly increased their action), and the second follows from the fact that $\ba'$ is a deviation from $\ba$ for a coalition $C$ with $I \subseteq C$ (since $\ba_I > \ba_I$).

We thus have $\bu(\ba'') > \bu(\ba)$, which contradicts the fact that $\ba$ is Pareto efficient.
Thus $I$ is empty and the lemma follows.
\end{proof}

Second:

\begin{lemma} \label{lem: path down dev}
Suppose there is a path $P \subset \rrp^n$ with endpoints $\bx, \by$ such that for any $\bp \in P$ we have $\bu(\bp) \le \bu(\ba)$.
If $\ba'$ is a deviation from $\ba$ for some coalition $C$ satisfying $\ba' \lneq \bx$, then $\ba'$ also must satisfy $\ba' \lneq \by$.
\end{lemma}
\begin{proof}
We walk along $P$ from $\bx$ towards $\by$ until we find the first point $\bp$ with $p_i = a'_i \ne 0$ for some $i$.
If we reach $\by$ before we find any such point $\bp$, it follows that $a'_i = 0$ or $a'_i < x_i$ for all $i$, and so $\ba' \lneq \bx$, as claimed.
Otherwise, we find such a point $\bp$, and we argue towards a contradiction.

We have $\bp \ne \ba'$, since $\bu_C(\ba') \gneq \bu_C(\ba)$ but $\bu_C(\bp) \le \bu_C(\ba)$.
By construction we then have $\bp \gneq \ba'$.
Since $p_i = a'_i$, by positive externalities we then have $u_i(p) > u_i(a')$.
Since $a'_i \ne 0$ we have $i \in C$, and since $\ba'$ is a deviation for $C$, this implies $u_i(a') > u_i(a)$.
We then have $u_i(p) > u_i(a)$, which contradicts the assumption that $\bu(\bp) \le \bu(\ba)$.
Therefore no such point $\bp$ may be found.
\end{proof}

We can now show:
\begin{proof} [Proof of Theorem \ref{thm: core classification}, Backwards Implication]
Assume that $\ba$ is robust to deviations by the grand coalition or any singleton coalition, and that $\bD_{\ba}$ is path connected.
Our goal is now to show that $\ba$ is in the core.

By Claim \ref{clm: vec adjust}, the property that the grand coalition has no deviation from $\ba$ implies that $\ba$ is Pareto efficient.
Thus, by Lemma \ref{lem: down dev} any deviation $\ba'$ from $\ba$ satisfies $\ba' \lneq \ba$.
Since no singleton coalition has a deviation from $\ba$ we have $\bz \in \bD_{\ba}$, and since $\bD_{\ba}$ is path connected there is a path contained in $\bD_{\ba}$ with endpoints $\ba, \ba$.
Thus, by Lemma \ref{lem: path down dev}, we further have that a deviation $\ba'$ must satisfy $\ba' \lneq \bz$.
Since no such point exists, it follows that no deviations from $\ba$ exist, and so $\ba$ is a core outcome.
\end{proof}

\subsection{Sketch of Forwards Implication of Theorem \ref{thm: core classification}}

We will denote by $\bd_{\bv} \bu(\ba)$ the one-sided directional derivative of $\bu$ at $\ba$ in the direction $\bv$.
In other words:
$$\bd_{\bv} \bu(\ba) := \lim_{\lambda \to 0^{+}} \frac{\bu(\ba + \lambda \bv) - \bu(\ba)}{\lambda}.$$
A nontrivial but standard fact from analysis is that, since $\bu$ is concave and well-defined everywhere, this limit is well-defined for all $\ba$, except when excluded by a boundary condition (e.g. if $v_i < 0$ but $a_i = 0$ for some agent $i$) -- see \cite{Lewin03}.


Our key lemma is:
\begin{lemma} \label{lem: deriv types}
At any outcome $\bz < \ba < \bo$, exactly one of the following three conditions holds:

\begin{enumerate}
\item There exist directions $\bv_{up} > \bz, \bv_{down} < \bz$ such that $\bd_{\bv_{up}} \, \bu(\ba) > \bz$ and $\bd_{\bv_{down}} \bu(\ba) < \bz$,

\item There exist directions $\bv_{up} > \bz, \bv_{down} < \bz$ such that $\bd_{\bv_{up}} \, \bu(\ba) \le \bz$ and $\bd_{\bv_{down}} \bu(\ba) \le \bz$, or

\item There exist directions $\bv_{up} > \bz, \bv_{down} < \bz$ such that $\bd_{\bv_{up}} \, \bu(\ba) < \bz$ and $\bd_{\bv_{down}} \bu(\ba) > \bz$.
\end{enumerate}
\end{lemma}

The three categories of Lemma \ref{lem: deriv types} carry a useful geometric intuition.
Specifically:
\begin{lemma} \label{lem: pf category}
The points in the second category of Lemma \ref{lem: deriv types} are precisely the Pareto Frontier.
\end{lemma}
The proofs of these two lemmas are quite technical, and can be found in Appendix \ref{app:class proof}.
With these in mind, we define
\begin{definition}
We will say that a point in the first category of Lemma \ref{lem: deriv types} is \emph{below the Pareto Frontier}, and a point in the third category of Lemma \ref{lem: deriv types} is \emph{above the Pareto Frontier} (and by Lemma \ref{lem: pf category}, the second category of points in Lemma \ref{lem: deriv types} are on the Pareto Frontier).
\end{definition}
The geometric intuition behind this definition is that, starting from a point $\bx$ in the first category, one can continuously follow the gradient $\bv_{up}$ to eventually obtain a Pareto efficient Pareto improvement $\bx' > \bx$ (we do not prove this fact formally; it is perhaps useful intuition but not essential to our main results).
Similarly, starting from $\bx$ in the third category, we can continuously follow the gradient $\bv_{down}$ to obtain a Pareto efficient Pareto improvement $\bx' < \bx$.

We then show:
\begin{proof} [Proof Sketch of Theorem \ref{thm: core classification}, Forwards Implication (Full proof in Appendix \ref{app:class proof})]
Suppose $\ba$ is a core outcome, and our goal is to show path connectedness of $\bD_{\ba}$.
First, we note that $\bz \in \bD_{\ba}$, since otherwise a singleton coalition can deviate from $\ba$ (and $\ba$ is in the core, so no such deviation is possible).
To show path connectedness of $\bD_{\ba}$, we consider an arbitrary point $\bx \in \bD_{\ba}$ and construct a path in $\bD_{\ba}$ from $\bx$ to $\bz$, thus implying that any two such points $\bx, \bx' \in \bD_{\ba}$ have a connecting path in $\bD_{\ba}$ via $\bz$.

We show the existence of the $\bx \leadsto \bz$ path with a careful repeated application of Lemma \ref{lem: deriv types}.
Informally speaking, we progressively slide $\bx > \bz$ a little bit closer to $\bz$ while maintaining the property $\bx \in \bD_{\ba}$.
If we ever hit $x_i = 0$ for some agent $i$, then we restrict our attention to the projected economy discluding agent $i$ and continue.
If we eventually exclude all agents in this manner, then we have $\bx = \bz$ and the process is complete.
Otherwise, suppose towards a contradiction that at some $\bx$, we cannot slide $\bx$ any closer to $\bz$ while maintaining $\bx \in \bD_{\ba}$.
We make two observations here: 
(1) $\bx$ must be above the Pareto frontier (else we could slide $\bx$ in the appropriate direction $\bv_{down}$) and so it belongs to the third; and (2) for all agents $i$ still being considered, we have $u_i(x) = u_i(a)$ (else, by the positive externalities assumption, we can unilaterally decrease the action of agent $i$ without destroying $\bx \in \bD_{\ba}$).
Hence, by moving $\bx$ slightly in the direction $\bv_{down}$ (which \emph{improves} the utility of all agents being considered), we have $u_i(x) > u_i(a)$ for all agents being considered, and so the new $\bx$ is a deviation from $\ba$.
Since we have assumed that $\ba$ is a core outcome, this is a contradiction, and so the process of sliding $\bx$ towards $\bz$ can never get stuck in this way.
\end{proof}

\subsection{Connection to Lindahl Equilibria}

Before proceeding towards our algorithm, we take a brief detour in this subsection to observe an interesting implication of Theorem \ref{thm: core classification} that helps illustrate its broader appeal.
Elliot and Golub \cite{EG13} show the following result:
\begin{theorem} [\cite{EG13}] \label{thm:eg}
The Lindahl equilibria of a public goods economy with a differentiable utility function are precisely the outcomes $\ba$ for which $\bd_{\ba} \bu(\ba) = \bz$. 
\end{theorem}
They phrase this theorem in different language related to the ``benefits network'' of the economy, but this formulation will suit our purposes better.
We refer the reader to their paper for a more in-depth discussion of the economic role of the Lindahl equilibria.

Combining Theorem \ref{thm:eg} with our machinery for Theorem \ref{thm: core classification}, we obtain:
\begin{theorem} \label{thm: local core}
In a public goods economy with a differentiable utility function, the Lindahl equilibria are precisely the core outcomes $\ba$ whose membership can be certified by examining only local information at $\ba$.
\end{theorem}
The proof of this theorem will use Theorem \ref{thm:eg} as a black box, and so we will not actually need to appeal to the formal definition of the Lindahl equilibria in its proof.
\begin{proof}
If $\ba$ is a Lindahl equilibrium, then by Theorem \ref{thm:eg} we have $\bd_{\ba} \, \bu(\ba) = \bd_{-\ba} \, \bu(\ba) = \bz$ (the first equality comes from the assumption that $\bu$ is differentiable).
We claim that any $\ba$ satisfying $\bd_{\ba} \, \bu(\ba) = \bd_{-\ba} \, \bu(\ba) = \bz$ is in the core, and thus its core membership can be verified by examining only these local derivatives.
First, note that $\ba$ is Pareto efficient by Lemma \ref{lem: pf category}.
Therefore, by Lemma \ref{lem: down dev}, any possible deviation $\ba'$ from $\ba$ satisfies $\ba' \lneq \ba$.
Now let $P$ be the line segment from $\bz$ to $\ba$.
By the assumption of concavity and the fact that $\bd_{-\ba} \, \bu(\ba) = \bz$, we have $\bu(\bp) \le \bu(\ba)$ for all $\bp \in P$.
Thus, by Lemma \ref{lem: path down dev}, we have $\ba' \lneq \bz$ and so $\ba'$ cannot exist and $\ba$ is a core outcome.

Now suppose that $\ba$ is not a Lindahl equilibrium, and so $\bd_{-\ba} \, \bu(\ba) \ne \bz$.
If we have $\bd_{-\ba} \, \bu(\ba) \lneq \bz$, then we have $\bd_{\ba} \, \bu(\ba) \gneq \bz$ (by differentiability) which implies that $\ba$ is not Pareto efficient, and hence is not a core outcome.
On the other hand, suppose we have $\bd_{-\ba} \, \bu(\ba) \gneq \bz$.
In this case, it is impossible to distinguish $\bu$ from the utility function $\bu'$ that is affine-linear everywhere and agrees with $\bu$ at $\ba$ using solely local information.
Note that $\ba$ is \emph{not} in the core of the economy defined by $\bu'$, since we have $u_i(\bz) > u_i(a)$ for whichever agent $i$ satisfies $d_{-a} \, u_i(\ba) > 0$.
Thus, if $\ba$ is in fact in the core of the economy defined by $\bu$, we will need to inspect non-local information about the economy to differentiate $\bu$ from $\bu'$.
\end{proof}
We note that it is possible to prove Theorem \ref{thm: local core} as a corollary directly from Theorem \ref{thm: core classification}, but this proof using the underlying machinery is simpler.

\section{Algorithm for Testing Core Membership}

Our main algorithmic result is:
\begin{theorem} \label{thm:algfull}
Given an outcome $\ba$ in a public goods economy, there is an algorithm (in the real-RAM model) that decides whether or not $\ba$ is in the core by solving $O(n)$ convex optimization problems and using $O(n)$ additional computation time.
\end{theorem}
%


The algorithm is fairly straightforward.
We maintain an ``active coalition'' $C_A$ throughout, as well as a proof that any agent $i \notin C_A$ must play action $a'_i = 0$ in any deviation $\ba'$ from $\ba$.
It is thus safe to assume that any deviating coalition $C$ satisfies $C \subseteq C_A$.
Initially $C_A \gets N$, so this invariant is trivially satisfied.
After each round, we either find a deviation for $C_A$ from $\ba$, or we remove one new agent from $C_A$.
Thus, if we make it $n$ rounds without finding a deviation, then we have $C_A = \emptyset$ and so no deviation from $\ba$ is possible.

\subsection{Preprocessing: Confirm Pareto Efficiency of $\ba$}

Before starting the main algorithm, we run the following two programs, with the purpose of testing whether or not $\ba$ is Pareto efficient.
\begin{program} \label{prg:uppe}
\begin{center}
Choose $\bv$ to maximize $\min \limits_{i} d_{v} \, u_i(a)$

Subj. to $\bv \ge \bz, \sum \limits_i v_i = 1$
\end{center}
\end{program}

\begin{program} \label{prg:downpe}
\begin{center}
Choose $\bv$ to maximize $\min \limits_{i} d_{v} \, u_i(a)$

Subj. to $\bv \le \bz, \sum \limits_i v_i = 1$
\end{center}
\end{program}
Note that the concavity of the optimized function $f(v) := \min_i d_v \, u_i(a)$ is immediate from the concavity of $\bu$.
By Lemma \ref{lem: deriv types}, we may immediately conclude that $\ba$ is not Pareto efficient (and thus not in the core) iff either of these programs optimizes at a point $\bv^*$ satisfying $f(v^*) > 0$.
Otherwise, we proceed with the knowledge that $\ba$ is Pareto efficient.
A key advantage of this is that, by Lemma \ref{lem: down dev}, we may now restrict our search for a deviation $\ba'$ to the bounded box $\bz \le \ba' \le \ba$.
This opens up the ability to use convex programming algorithms, which typically require bounded domains, in the remainder of the algorithm.\footnote{This detail is precisely why we use Programs \ref{prg:uppe} and \ref{prg:downpe} to check the Pareto efficiency of $\ba$, rather than the ostensibly simpler method of searching for $\bx^*$ that maximizes $\min_i u_i(x^*) - u_i(a)$: the latter method requires a search for $\bx^*$ over an unbounded search space, which rules out many popular methods of convex optimization that we wish to keep available.}

\subsection{Main Loop: Shrinking $C_A$}

Each of the $n$ rounds of the algorithm consists of three steps.
First, we restrict our attention to the projected economy for the coalition $C_A$.
Second, we run the following program:

\begin{program} \label{prg:PE}
\begin{center}
Choose $\bx$ to maximize $\min \limits_{i}  u_i(x) - u_i(a)$

Subj. to $\bz \le \bx \le \ba_{C_A}$
\end{center}
\end{program}

Let $\bx^*$ be a maximizing point of Program \ref{prg:PE}.
We have:
\begin{lemma} \label{lem:deriv-quality}
Either $\bu^{C_A}(\bx^*) > \bu_{C_A}(\ba)$ or $\bu^{C_A}(\bx^*) \le \bu_{C_A}(\ba)$, and $\bx^*$ is Pareto efficient.\footnote{Note that these statements hold specifically in the projected economy for $C_A$.}
\end{lemma}
\begin{proof}
First we argue Pareto efficiency.
If $\bx'$ is a Pareto improvement on $\bx^*$, then by Claim \ref{clm: vec adjust} there is another point $\bx''$ with $\bu(\bx'') > \bu(x^*)$.
This $\bx''$ would be a superior maximizing point for Program \ref{prg:PE}, so there can be no Pareto improvement on $\bx^*$.

Next, let $i := \arg \max_i u_i(x^*) - u_i(a)$ and $j := \arg \min_j u_j(x^*) - u_j(a)$.
If $u_i(x^*) - u_i(a) > u_j(x^*) - u_j(a)$, then (by the same argument used in Claim \ref{clm: vec adjust}) we can again obtain a superior maximizing point $\bx^{**}$ by slightly increasing the action of agent $i$ from $\bx^*$.
Thus we have $u_i(x^*) - u_i(a) = u_j(x^*) - u_j(a)$, and it follows that either $\bu(\bx^*) > \bu(\ba)$ or $\bu(\bx^*) \le \bu(\ba)$.
\end{proof}

In the former case where $\bu(\bx^*) > \bu(\ba)$, it follows that $\bx^*$ is a deviation from $\ba$ for the coalition $C_A$, so we may halt the algorithm.
Otherwise, we have $\bu(\bx^*) \le \bu(\ba)$.
We then observe:

\begin{lemma} \label{lem:proj down dev}
If $\bu^{C_A}(\bx^*) \le \bu_{C_A}(\ba)$, then in the full (non-projected) economy, any deviation $\ba'$ from $\ba$ for a coalition $C \subseteq C_A$ satisfies $\ba'_{C_A} \lneq \bx^*$.
\end{lemma}
\begin{proof}
The deviation $\ba'$ satisfies $\ba'_{N \setminus C} = \bz$ and $\bu_{C}(\ba') > \bu_C(\ba) \ge \bu_C(\bx^*)$.
It follows that $\ba'_C$ is also a deviation for $C$ from $\bx^*$ in the projected economy for $C_A$.
The claim is then immediate from Lemma \ref{lem: down dev}.
\end{proof}
One step remains.
We run:
\begin{program} \label{prg:downdir}
\begin{center}
Choose $\bv$ to maximize $\min \limits_i d_v \, u_i(x^*)$

Subj. to $\bv \le \bz, \sum \limits_i v_i = 1$
\end{center}
\end{program}

Let $\bv^*$ be a point that maximizes Program \ref{prg:downdir}.
We have
\begin{lemma} \label{lem:down deriv neg}
$$\bd_{\bv^*} \, \bu(\bx^*) \le \bz$$
\end{lemma}
The proof is very similar to the proof of Lemma \ref{lem:deriv-quality}, so we omit it for now.
We then finally have:
\begin{lemma}
Let $i := \arg \min_i x^*_i / v^*_i$.
Then any deviation $\ba'$ from $\ba$ has $a'_i = 0$.
\end{lemma}
\begin{proof}
By Lemma \ref{lem:proj down dev}, we have $\ba'_{C_A} \lneq \bx^*$.
Let $P$ be the line segment starting at $\bx^*$, extending in the direction $\bv^*$ until a point $\bp^*$ is reached where $p^*_i = 0$ for some agent $i$; note that this will specifically be $i = \arg \min_i x^*_i / v^*_i$.
By concavity, all $\bp \in P$ satisfies $\bu(\bp) \le \bu(\bx^*) \le \bu(\ba)$.
Noting once again that $\ba'_{C_A}$ is a deviation for $C$ from $\bx^*$ in the projected economy for $C_A$, it follows from Lemma \ref{lem: path down dev} that $\ba'_{C_A} \lneq \bp^*$, and so $a'_i = 0$.
\end{proof}

With this in mind, the final step in the loop is to delete $i$ from $C_A$ and repeat.
After $n$ repetitions, we have $C_A = \emptyset$, so we may halt the algorithm and report that $\ba$ is in the core.

\subsection{Algorithm Pseudocode}

To recap the algorithm, which has been interspersed with proofs of correctness above, we give full pseudocode here.

\begin{algorithm}
Let $\bv^*_1 \gets \text{output of Program } \ref{prg:uppe}$\;
Let $\bv^*_2 \gets \text{output of Program } \ref{prg:downpe}$\;
\If{$\bd_{\bv^*_1} \, \bu(\ba) > \bz$ or $\bd_{\bv^*_1} \, \bu(\ba) > \bz$}{
\Return{``$\ba$ is not in the core''};
}

$C_A \gets N$\;
\While{$C_A \ne \emptyset$}{
$\bx^* \gets \text{output of Program } \ref{prg:PE}$\;
\If{$\bu^{C_A}(\bx^*) > \bu_{C_A}(\ba)$}{
\Return{``$\ba$ is not in the core''};
}
$\bv^* \gets \text{output of Program } \ref{prg:downdir}$\;
$C_A \gets C_A \setminus \left\{\arg\min_{i \in C_A} x^*_i / v^*_i \right\}$\;
}
\Return{``$\ba$ is in the core''};

\caption{Testing Core Membership of $\ba$}
\end{algorithm}

\subsection{Conclusion: Adapting the Algorithm for Approximate Optimization}

Our algorithm implies that core membership testing is efficient under any utility function that admits quick solving of convex programs as described above.
However, it may still be desirable to test core membership as best as possible when the underlying utility function is either unknown or badly behaved and so exact convex optimization is impossible.
Our algorithm can indeed be adapted to this effect, with a few significant points of caution, by substituting in modern approximate optimization algorithms.
Due to space constraints, we include a discussion of this adaptation in Appendix \ref{app:approx}.

\section{Acknowledgements}

I am very grateful to Ben Golub, who introduced me to this area of game theory and many of the problems addressed in this paper, and who helped advise an early version of this research project.
I am also grateful to Ben Hescott for mentorship and writing advice during the early stages of this project.
Finally, I thank an anonymous reviewer for useful criticism and corrections on an earlier draft of this paper.

\bibliography{workBIB}

\begin{thebibliography}{10}

\bibitem{Allouch13}
N.~Allouch.
\newblock The cost of segregation in social networks.
\newblock {\em Queen Mary Working Paper}, 2013.

\bibitem{Allouch15}
N.~Allouch.
\newblock On the private provision of public goods on networks.
\newblock {\em Journal of Economic Theory}, pages 527--552, 2015.

\bibitem{BCZ06}
C.~Ballester, A.~Calv{\' o}-Armengol, and Y.~Zenou.
\newblock Who’s who in networks. wanted: The key player.
\newblock {\em Econometrica}, 74:1403--1417, 2006.

\bibitem{BV04}
Stephen Boyd and Lieven Vandenberghe.
\newblock {\em Convex Optimization}.
\newblock Cambridge University Press, 2004.

\bibitem{BK07}
Y.~Bramoulle and R.~Kranton.
\newblock Public goods in networks.
\newblock {\em Journal of Economic Theory}, 135:478--494, 2007.

\bibitem{Bubeck15}
S.~Bubeck.
\newblock Convex optimization: Algorithms and complexity.
\newblock {\em Foundations and Trends in Machine Learning}, 8(3-4):231--358,
  2015.

\bibitem{Cameron94}
Peter Cameron.
\newblock {\em Combinatorics: Topics, Techniques, Algorithms}.
\newblock Cambridge University Press, 1994.

\bibitem{CT97}
Parkash Chander and Henry Tulkens.
\newblock The core of an economy with multilateral environmental externalities.
\newblock {\em International Journal of Game Theory}, 26:379--401, 1997.

\bibitem{CS04}
V.~Conitzer and T.~Sandholm.
\newblock Computing shapley values, manipulating value division schemes, and
  checking core membership in multi-issue domains.
\newblock In {\em Proc. National Conference on Artificial Intelligence (AAAI)},
  pages 219--225, 2004.

\bibitem{CS06}
V.~Conitzer and T.~Sandholm.
\newblock Complexity of constructing solutions in the core based on synergies
  among coalitions.
\newblock {\em Artificial Intelligence}, pages 607--–619, 2006.

\bibitem{Curriel97}
Imma Curriel.
\newblock {\em Cooperative Game Theory and Applications}.
\newblock Springer US, 1 edition, 1997.

\bibitem{DP94}
X.~Deng and C.~Papadimitriou.
\newblock On the complexity of cooperative solution concepts.
\newblock {\em Mathematics of Operations Research}, 19:257, 1994.

\bibitem{Edgeworth81}
F.~Y. Edgeworth.
\newblock Mathematical psychics: An essay on the mathematics to the moral
  sciences.
\newblock {\em Reprinted in Diamond, M.A.}, 1881.

\bibitem{EGGW09}
E.~Elkind, L.~A. Goldberg, P.~W. Goldberg, and M.~Wooldridge.
\newblock A tractable and expressive class of marginal contribution nets and
  its applications.
\newblock {\em Math. Log. Q.}, 2009.

\bibitem{EG13}
Matthew Elliot and Benjamin Golub.
\newblock {A network approach to public goods}.
\newblock In {\em Proc. 14th Electronic Commerce (EC)}, pages 377--378, 2013.

\bibitem{FFHK97}
U.~Faigle, S.~Fekete, W.~Hochstattler, and W.~Kern.
\newblock On the complexity of testing membership in the core of min-cost
  spanning trees.
\newblock {\em Internat. J. Game Theory}, 26:361--–366, 1997.

\bibitem{Foley70}
D.~K. Foley.
\newblock Lindahl's solution and the core of an economy with public goods.
\newblock {\em Econometrica}, 38:66--72, 1970.

\bibitem{Gillies59}
D.~B. Gillies.
\newblock Solutions to general non-zero-sum games.
\newblock {\em Contributions to the Theory of Games IV}, pages 47--85, 1959.

\bibitem{GS00}
M.~Goemans and M.~Skutella.
\newblock Cooperative facility location games.
\newblock In {\em Symposium on Discrete Algorithms (SODA)}, pages 76--85, 2000.

\bibitem{GolubPC}
Benjamin Golub.
\newblock Personal correspondence, 2012.

\bibitem{GMPS11}
G.~Greco, E.~Malizia, L.~Palopoli, and F.~Scarcello.
\newblock On the complexity of the core over coalition structures.
\newblock In {\em International Joint Conference on Artificial Intelligence
  (IJCAI)}, 2011.

\bibitem{IS05}
S.~Ieong and Y.~Shoham.
\newblock Marginal contribution nets: A compact representation scheme for
  coalitional games.
\newblock In {\em Electronic Commerce (EC)}, pages 193--202, 2005.

\bibitem{KPRST09}
S.~Kintali, L.~Poplawski, R.~Rajaraman, R.~Sundaram, and S.~Teng.
\newblock Reducibility among fractional stability problems.
\newblock In {\em Proc. 50th FOCS}, pages 283--292, 2009.

\bibitem{Lewin03}
Jonathan Lewin.
\newblock {\em An interactive introduction to mathematical analysis}.
\newblock Cambridge University Press, 2003.

\bibitem{LC14}
Y.~Li and V.~Conitzer.
\newblock Complexity of stability-based solution concepts in multi-issue and
  mc-net cooperative games.
\newblock In {\em Proc. 2014 international conference on Autonomous agents and
  multi-agent systems (AAMAS)}. International Foundation for Autonomous Agents
  and Multiagent Systems, 2014.

\bibitem{PS07}
Bezalel Peleg and Peter Sudh{\" o}lter.
\newblock {\em Introduction to the Theory of Cooperative Games}.
\newblock Springer, 2 edition, 10 2007.

\bibitem{Protasov96}
V.~Protasov.
\newblock Algorithms for approximate calculation of the minimum of a convex
  function from its values.
\newblock {\em Mathematical Notes}, 1:69--74, 1996.

\bibitem{Samuelson54}
Paul Samuelson.
\newblock The pure theory of public expenditure.
\newblock {\em Review of Economics and Statistics}, 37(4):350--356, 1954.

\bibitem{Scarf67a}
H.~Scarf.
\newblock The core of an n person game.
\newblock {\em Econometrica}, 1967.

\bibitem{SS66}
L.S. Shapley and M.~Shubik.
\newblock Quasi-cores in a monetary economy with non-convex preferences.
\newblock {\em Econometrica}, 34:805–827, 1966.

\bibitem{SD07}
S.~Sung and D.~Dimitrov.
\newblock On core membership testing for hedonic coalition formation games.
\newblock {\em Operations Research Letters}, 35:155--158, 2007.

\end{thebibliography}
	\bibliographystyle{plain}

\appendix
\section{Remaining Proof of Theorem \ref{thm: core classification} \label{app:class proof}}

\subsection{Proof of Lemmas \ref{lem: deriv types} and \ref{lem: pf category}}

We begin with some useful well-known technical claims.
Proofs can be found in most textbooks on analysis and/or combinatorics, see e.g. \cite{Lewin03, Cameron94}.
\begin{claim} \label{clm: dir derivs}
For any fixed $\ba > \bz$, the function $\beff(\bv) := \bd_{\bv} \bu(\ba)$ takes values in $\rr^n$ for all $\bv$ (i.e. the limit is well-defined and not infinite on all indices) and $\beff$ is concave (which follows from the assumption that $\bu$ is concave).
\end{claim}

\begin{claim} \label{clm: lipschitz}
Any concave function $\beff : \rr^n \to \rr$ with domain restricted to a bounded subset of $\rr_n$ is Lipschitz continuous.
That is, for any bounded domain $D \subset \rr_n$, there is a constant $K_D$ such that for any $\bv_1, \bv_2 \in \rr^n$ we have $\| \beff(\bv_1) - \beff(\bv_2) \| \le K_D \cdot \|\bv_1 - \bv_2\|$.
\end{claim}

\begin{lemma} [Sperner's Lemma] \label{lem: sperner}
Let $S$ be an $(n-1)$-simplex.\footnote{An $(n-1)$-simplex is an $n-1$ dimensional analogue of a triangle defined by $n$ vertices; e.g. a $2$-simplex is a triangle, a $3$-simplex is a tetrahedron, etc.}
Let $\tee$ be a partition of $S$ into smaller $(n-1)$-simplices (a ``simplicization'' of $S$), and color each point in $S$ that is a vertex of some sub-simplex $T \in \tee$ by one of $n$ colors.
If there is no vertex on the face of $S$ opposite its $i^{th}$ vertex that is given color $i$, then there exists a ``rainbow'' simplex $T \in \tee$ for which the $n$ vertices of $T$ have all $n$ possible colors.
\end{lemma}

\begin{lemma} [Cantor's Intersection Theorem] \label{lem: cantor}
Let $\{X_i\}$ be an infinite sequence of nonempty sets that are closed and bounded, with $X_i \supseteq X_{i+1}$ for all $i$.
Then
$\bigcap \limits_i X_i \ne \emptyset$.
\end{lemma}

We now proceed with some original observations:

\begin{claim} \label{clm: up or down}
For any outcome $\ba > \bz$, if there exists a direction $\bv > \bz$ ($\bv < \bz$) for which $\bd_{\bv} \bu(\ba) \gneq \bz$, then there is no direction $\bv' > \bz$ ($\bv' < \bz$) for which $\bd_{\bv'} \bu(\ba) \le \bz$.
\end{claim}
\begin{proof}
We will prove the claim for $\bv, \bv' > \bz$; the case where $\bv, \bv' < \bz$ follows from an identical argument.

Let $\bv > \bz$ such that $\bd_{\bv} \bu(\ba) \gneq \bz$, and let $\bv' > \bz$ be an arbitrary direction that is not a scalar multiple of $\bv$.
We can then find a vector $\bx > \bz$ for which $\bv + \bx = \lambda \bv'$ for some scalar $\lambda > 0$, such that $x_i = 0$ for some $i$.
We have that $(d_v u(a))_i > 0$ (by hypothesis) and $(d_x u(a))_i > 0$ (by the positive externalities assumption).
Thus, by concavity of $\bu$, we have $(d_{v + x} u(a))_i > 0$, and so $(d_{\lambda v'} u(a))_i > 0$.
Therefore we cannot have $\bd_{\bv'} \bu(\ba) \le \bz$.
\end{proof}

Now we are ready to show:
\begin{lemma} [Lemma \ref{lem: deriv types} in the body of the paper]
At any outcome $\ba > \bz$, exactly one of the following three conditions holds:

\begin{enumerate}
\item There exist directions $\bv_{up} > \bz, \bv_{down} < \bz$ such that $\bd_{\bv_{up}} \, \bu(\ba) > \bz$ and $\bd_{\bv_{down}} \bu(\ba) < \bz$,

\item There exist directions $\bv_{up} > \bz, \bv_{down} < \bz$ such that $\bd_{\bv_{up}} \, \bu(\ba) \le \bz$ and $\bd_{\bv_{down}} \bu(\ba) \le \bz$, or

\item There exist directions $\bv_{up} > \bz, \bv_{down} < \bz$ such that $\bd_{\bv_{up}} \, \bu(\ba) < \bz$ and $\bd_{\bv_{down}} \bu(\ba) > \bz$.
\end{enumerate}
\end{lemma}
\begin{proof}
First, we observe that Claim \ref{clm: up or down} implies that no $\ba$ can satisfy two of these conditions simultaneously: if $\ba$ satisfies the first condition then there is no $\bv_{up}$ satisfying the latter two conditions; if $\ba$ satisfies the last condition then there is no $\bv_{down}$ satisfying either of the former two conditions.
Thus, it suffices to show that each point $\ba$ satisfies \emph{at least} one of these three conditions, and we immediately have that it satisfies \emph{exactly} one of these three conditions.

Second, by Claim \ref{clm: vec adjust}, it suffices to show the existence of $\bv_{up} \gneq \bz, \bv_{down} \lneq \bz$ satisfying these conditions.

Now we begin the main proof.
Let $S$ be the $(n-1)$-simplex in $\rr^n$ with vertex set equal to the $n$ basis vectors $\{\be_i\}$ (where $\be_i \in \rr^n$ is the vector with a $1$ in the $i^{th}$ place and a $0$ elsewhere).
If $\min_i (d_s \, u(a))_i > 0$ for any $\bs \in S$, then we have $\bd_{\bs} \, \bu(\ba) > \bz$.
By convex preferences we also have that $\bv_{down} := -\bv_{up}$ satisfies $\bd_{\bv_{down}} \bu(\ba) < \bz$, and so we may place $\ba$ into the first category of Lemma \ref{lem: deriv types}.

Otherwise, assume that $\min_i (d_s \, u (a))_i \le 0$ for all $\bs \in S$.
In this case, our goal is to show that there is $\bv_{up}$ with $\bd_{\bv_{up}} \bu(\ba) \le \bz$.
Assign each agent a color, and color each point $\bs \in S$ by the agent $\arg\min_i (d_s \, u(a))_i$, breaking ties arbitrarily.
Let $F_i$ be the face of $S$ opposite $\be_i$; that is, $F_i$ is the set of $\bs \in S$ for which $s_i = 0$.
Let $\beff$ be any point in $F_i$.
Note that, by the assumption of positive externalities, we have $(d_f \, u(a))_i > 0$; since we have assumed that $\min_i (d_f \, u (a))_i \le 0$, this implies that $\beff$ is not colored by agent $i$.

Let $\tee$ be a simplicization of $S$ such that each simplex $T \in \tee$ fits in a ball of radius $\eps$, for some $\eps > 0$ that we will pick later.
Since the face $F_i$ does not include any points with color $i$, we may apply Sperner's Lemma (Lemma \ref{lem: sperner}) to conclude that there is a simplex $T \in \tee$ whose $n$ vertices $\{\bt^1, \dots, \bt^n\}$ are rainbow (we assume without loss of generality that vertex $\bt^i$ has color $i$).
Let $\bx$ be a point within distance $\eps$ of each vertex $\bt^i$.
Since $(d_{t^i} u(a))_i \le 0$, and since $\beff(\bd) = \bd_{\bv} \bu(\ba)$ is Lipschitz continuous on $S$ (from Claims \ref{clm: dir derivs} and \ref{clm: lipschitz}), it follows that there is a constant $K_S$ independent of $\eps$ such that $(d_x u(a))_i \le \eps K_S$ for all $i$.
In other words, we have that $\bu(\bx) \le \eps K_S \bo$, and so the set
$$Z_{\lambda} := \{\bs \in S \ \mid \ \bd_{\bs} \bu(\ba) \le \lambda \bo \}$$
is nonempty for all $\lambda > 0$, since it contains $\bx$ when we choose $\eps \le \frac{\lambda}{K_S}$.

Since the function $f(\bs) = \bd_{\bs} \bu(\ba)$ is continuous, $Z_{\lambda}$ for any $\lambda > 0$ is closed, bounded, and nonempty.
Thus, if we express the set
$$Z_0 := \{\bs \in S \ \mid \ \bd_{\bs} \bu(\ba) \le \bz\} \qquad  \text{    as    } \qquad Z_0 = \bigcap \limits_{\lambda > 0} Z_{\lambda}$$
then we see that $Z_0$ is the intersection of a descending sequence of sets that are closed, bounded, and nonempty.
By Cantor's Intersection Theorem (Lemma \ref{lem: cantor}), $Z_0$ is nonempty.
Note that any $\bv \in Z_0$ satisfies $\bv > \bz$, since if $v_i = 0$ then by the positive externalities assumption we have $(d_v u(a))_i \le 0$.
We may thus take $\bv_{up}$ to be any element of $Z_0$, and we have $\bv_{up} > 0$ and $\bd_{\bv_{up}} \bu(\ba) \le \bz$.

We have now shown that there is a vector $\bv_{up} > \bz$ that satisfies either $\bd_{\bv_{up}} \bu(\ba) > \bz$ or $\bd_{\bv_{up}} \bu(\ba) \le \bz$.
By an exactly symmetric argument, we have that there is a vector $\bv_{down} < \bz$ that satisfies either $\bd_{\bv_{down}} \bu(\ba) > \bz$ or $\bd_{\bv_{down}} \bu(\ba) \le \bz$.
Thus, if $\bd_{\bv_{up}} \bu(\ba) > \bz$ then we may place $\ba$ into the first category, if $\bd_{\bv_{down}} \bu(\ba) > \bz$ then we may place $\ba$ into the last category, and if neither of these are true then we simultaneously have $\bd_{\bv_{up}} \bu(\ba) \le \bz$ and $\bd_{\bv_{down}} \bu(\ba) \le \bz$ and so we may place $\ba$ into the second category.
\end{proof}

\begin{claim} \label{clm: hilo pf}
If there is a Pareto improvement $\ba'$ on an outcome $\ba$ satisfying $\ba' \not \lneq \ba$, then there is a Pareto improvement $\ba''$ on $\ba$ satisfying $\ba'' > \ba$.
\end{claim}
\begin{proof}
Let $C$ be the largest set of agents for which $\ba'_C > \ba$.
If $C$ is empty then we have $\ba' \lneq \ba$ and we are done.
Otherwise, suppose $C$ is nonempty.
Define the point $\ba^{\max}$ such that $a^{\max}_j := \max\{a'_i, a_i\}$ for all $j$.
By positive externalities, we have $\bu_C(a^{\max}) \ge \bu_C(\ba') \ge \bu_C(\ba)$ (where the latter inequality follows from the fact that $\ba'$ is a Pareto improvement on $\ba$).
Also by positive externalities, we have $\bu_{N \setminus C}(\ba^{\max}) > \bu_{N \setminus C}(\ba')$.
Thus, $\ba^{\max}$ is a Pareto improvement on $\ba$ which satisfies $\ba^{\max} \gneq \ba$.
By Claim \ref{clm: vec adjust}, we may further assume that $\ba^{\max} > \ba$.
\end{proof}

\begin{lemma} [Lemma \ref{lem: pf category} in the body of the paper]
The points in the second category of Lemma \ref{lem: deriv types} are precisely the Pareto Frontier.
\end{lemma}
\begin{proof}
It is clear that no outcome $\ba$ in the first or third category of Lemma \ref{lem: deriv types} is Pareto efficient, since these points have a direction $\bv$ for which $\bd_{\bv} \bu(\ba) > \bz$.
Therefore the second category in Lemma \ref{lem: deriv types} contains the Pareto frontier.
Let $\ba$ be an outcome in the second category, and suppose towards a contradiction that $\ba'$ is a Pareto improvement on $\ba$.

By Claim \ref{clm: hilo pf}, we may assume that $\ba' \lneq \ba$ or $\ba' \gneq \ba$; by Claim \ref{clm: vec adjust} we may assume more strongly that $\ba > \ba'$ or $\ba < \ba'$.
By concavity of $\bu$, we have that $\bd_{\ba' - \ba} \bu(\ba) > \bz$.
There also exist directions $\bv_{up}, \bv_{down}$ for which $\bd_{\bv_{up}} \bu(\ba) \le \bz$ and $\bd_{\bv_{down}} \bu(\ba) \le \bz$.
Since $\ba' - \ba > \bz$ or $\ba' - \ba < \bz$, we have a contradiction to Claim \ref{clm: up or down}.
Thus, there can be no Pareto improvement $\ba'$ on $\ba$.
\end{proof}

\subsection{Proof of Theorem \ref{thm: core classification}}

We assume that $\ba$ is a core outcome.
By definition this implies that the grand coalition and singleton coalitions have no deviation from $\ba$.
Thus, it only remains to show that $\bD_{\ba}$ is path connected.

\begin{claim} \label{clm: le points}
Let $\bx$ be a point on the lower envelope of a path connected component of $\bD_{\ba}$ (i.e. there is no point $\bx' \lneq \bx$ in this component).\footnote{The existence of $\bx$ follows straightforwardly from the fact that the path connected components of $\bD_{\ba}$ are topologically closed (since $\bu$ is continuous).}
Then for some (possibly empty) set of agents $C \subseteq N$, we have $\bu_C(\bx) = \bu_C(\ba)$ and $\bx_{N \setminus C} = \bz_{N \setminus C}$.
\end{claim}
\begin{proof}
Define $C$ to be the set of largest agents for which $\bx_C > \bz_C$.
By definition we have $\bx_{N \setminus C} = \bz_{N \setminus C}$, so we only need to prove that $\bu_C(\bx) = \bu_C(\ba)$. 

Suppose towards a contradiction that there is an agent $i$ for whom $x_i > 0$ and $u_i(x) < u_i(a)$.
We can then modify $\bx$ to a new point $\bx'$ by slightly decreasing the action of agent $i$.
Since $\bu$ is continuous, if we make this decrease small enough then we still have $u_i(x') < u_i(a)$.
By positive externalities, we then have $\bu(\bx') < \bu(\ba)$.
Since $\bx' \lneq \bx$, this contradicts the fact that $\bx$ is on the lower envelope of its connected component.
\end{proof}

\begin{lemma} \label{lem: da connected}
If $\ba$ is a core outcome, then $\bD_{\ba}$ is path connected.
\end{lemma}
\begin{proof}
We prove the contrapositive.
Suppose $\bD_{\ba}$ is not path connected, and let $\bar{\bD}_{\ba}$ be a path connected component of $\bD_{\ba}$ that does not include $\bz$, and let $\bx$ be a point on the lower envelope of $\bar{\bD}_{\ba}$.
Let $C$ be the smallest coalition such that $\bu_C(\bx) = \bu_C(\ba)$ and $\bx_{N \setminus C} = \bz_{N \setminus C}$ (the existence of such a coalition is given by Claim \ref{clm: le points}); note that $\bx \ne \bz$ and so $C$ is nonempty.
Now consider the projected economy for $C$.
We have $\bx_C > \bz_C$, so Lemma \ref{lem: deriv types} applies to the outcome $\bx_C$ in this economy.

Our next step is to figure out which of the three categories in Lemma \ref{lem: deriv types} contains $\bx_C$ in the projected economy.
We observe that $\bx_C$ cannot fall into the first or second category of Lemma \ref{lem: deriv types}.
This holds because we would have $\bv_{down} < \bz$ with $\bd_{\bv_{down}} \bu(\ba) \le \bz$, which means that we can slightly decrease the actions of the coalition $C$ in the \emph{original} economy along the direction $\bv_{down}$, giving a new point $\bx' \lneq \bx$.
By concavity of $\bu$, we still have $\bu(\bx') \le \bu(\ba)$ and so $\bx'$ belongs to the same connected component of $\bar{D}_{\ba}$; this contradicts the assumption that $\bx$ is on the lower envelope of $\bar{D}_{\ba}$.
Thus $\bx_C$ is in the third category of Lemma \ref{lem: deriv types} in the projected economy, so there is a direction $\bv_{down} < \bz$ with $\bd_{\bv_{down}} \bu(\ba) > \bz$.
If we move $\bx_C$ slightly in this direction $\bv_{down}$, we obtain an outcome $\bx'_C$ with $\bu_C(\bx'_C) > \bu_C(\ba_C)$.
By setting $\bx'_{N \setminus C} = \bz$, we then concatenate $\bx'_C$ and $\bx'_{N \setminus C}$ to obtain a point $\bx'$, which is a deviation in the original economy for the coalition $C$.
\end{proof}

\section{Adapting the Algorithm for Approximate Optimization \label{app:approx}}

We conclude by discussing the feasibility of the algorithm on utility functions that are either unknown or badly behaved, where exact convex optimization might not be easily solvable.
Here, one can use modern approximate convex optimization algorithms in the place of the exact solutions assumed in the above algorithm.
These algorithms query the utility function and its derivatives at various points in the domain, ultimately finding a point that comes within an additive $\eps$ of the optimal value.
Their theoretical runtime is typically a large polynomial in the dimensionality of the search space ($n$) but quite good in the accuracy parameter ($\eps$); for example, randomized center-of-gravity methods use $\Oish(n^5 \log(\eps^{-1}))$ time \cite{Bubeck15}, leading to a runtime of $\Oish(n^6 \log(\eps^{-1}))$ for our algorithm.
Further improvements are possible via more specialized optimization algorithms if one wishes to assume nice properties for the utility function.
If one does not wish to assume that the utility function derivatives can be easily measured, these algorithms can be adapted to avoid this necessity, although at the cost of an additional factor of $n$ \cite{Protasov96}.
For a general reference on convex optimization, see \cite{BV04}.

Let us consider a natural relaxation in which the goal is to distinguish between core outcomes and outcomes that are $\eps$-far from being in the core (i.e. there is an ``$\eps$-deviation'' $\ba'$ from $\ba$ satisfying $\bu_C(\ba') > \bu_C(\ba) + \eps \bo$ for the deviating coalition $C$).
Two tweaks to the algorithm/model will be needed to solve this relaxed problem with an approximate convex optimization algorithm.

\subsection{First Change: A Bounded Action Space}
First, caution is needed in the preprocessing stage where the Pareto efficiency of $\ba$ is tested.
It is no longer feasible to use the derivative-testing technique employed here in Programs \ref{prg:uppe} and \ref{prg:downpe}: the unbounded domain means that even a \emph{slightly} positive entry in the directional derivatives can translate to an indefinitely large Pareto improvement by following this gradient out far enough.
In other words, approximately negative directional derivatives do not translate to approximately Pareto efficient points.
Perhaps the most natural workaround here is to simply enforce a bounded domain, i.e. assume that actions lie in the interval $[0, 1]$ rather than $[0, \infty)$.
It seems fairly economically natural to assume a maximum amount of possible work for each agent.

With this assumption in hand, we can actually skip the entire preprocessing phase and move straight to the main loop where we iteratively shrink $C_A$.
The only purpose of the processing phase in the original algorithm is to bound the deviation search space via Lemma \ref{lem: down dev}; if we assume a bounded action space in the first place, it is no longer needed.

\subsection{Second Change: Dependence on the Condition Number}

The second change needed in the algorithm is that we must incurr an additional dependence in runtime on the condition number of the economy
$$\kappa := \max_{v, x, i} d_v \, u_i(x).$$
This dependence comes from the interplay between Programs \ref{prg:PE} and \ref{prg:downdir}.
To demonstrate the issue, suppose we run Program \ref{prg:PE} up to accuracy $\eps$, returning a point $\widehat{\bx^*}$.
We can easily examine $\widehat{\bx^*}$ to distinguish the cases in which the current active coalition $C_A$ has an $\eps$-deviation versus no deviation at all.
It is also quite straightforward to extend Lemmas \ref{lem: down dev} and \ref{lem: path down dev}.
The trouble is that Lemma \ref{lem:down deriv neg} breaks horribly; in particular, the vector $\bv^* \le \bz$ returned by Program \ref{prg:downdir} might not even approximately satisfy $\bd_{\bv^*} \, \bu(\widehat{\bx^*}) \le \bz$.

We propose the following solution.
First, we improve the accuracy of Program \ref{prg:PE} from $\eps$ to $\eps/(2\kappa)$, and again let $\widehat{\bx^*}$ be its output.
Thus if $\widehat{\bx^*}$ is not a deviation, then there is no $\eps/(2\kappa)$-deviation from $\ba$ for $C_A$.
Second, we begin an iterative process where we repeatedly select an agent $i$ where $u_i(\widehat{\bx^*}) \le u_i(a) + \eps/2$, and we modify $\widehat{\bx^*}$ by decreasing the action of agent $i$ by $\eps/(2\kappa)$.
Note that we now (quite generously) have $u_i(\widehat{\bx^*}) \le u_i(a) + 2\eps/3$, and we may repeat the process on a new agent.
In the worst case, we may need to repeat this process $O(n\kappa / \eps)$ times before an agent's action reaches $0$.
However, we can halt the process as soon as we slide $\widehat{\bx^*}$ to a point below the Pareto frontier, at which point we may simply compute the appropriate downwards derivative by Program \ref{prg:downdir}.
We were unable to determine if there actually exist utility functions in which the additional dependence of $O(n\kappa / \eps)$ will actually be realized in most iterations of the algorithm.

\end{document}